\documentclass[fleqn, authoryear]{article}
\usepackage{authblk, mathtools, geometry, amsthm, enumitem, lmodern, txfonts, natbib, doi, url}
\usepackage{etoolbox}\apptocmd{\sloppy}{\hbadness 10000\relax}{}{}
\newtheorem{theorem}{theorem}[section]
\newtheorem{lemma}[theorem]{Lemma}
\newtheorem{proposition}[theorem]{Proposition}

\newtheorem{definition}[theorem]{Definition}
\theoremstyle{definition}\newtheorem{example}[theorem]{Example}
\DeclareMathOperator*{\argmax}{arg\,max}
\DeclareMathOperator*{\conv}{conv}
\newcommand{\bri}{\operatorname{BR^{I}}}\newcommand{\brii}{\operatorname{BR^{II}}}\newcommand{\I}{$\mathrm{I}$ }
\newcommand{\gi}{\operatorname{\Gamma^{I}}}\newcommand{\gii}{\operatorname{\Gamma^{II}}}\newcommand{\II}{$\mathrm{II}$ }
\renewcommand{\(}{\left(}\renewcommand{\)}{\right)}

\title{On the commitment value and commitment optimal strategies in bimatrix games}
\author[1]{Stefanos Leonardos\footnote{Stefanos Leonardos gratefully acknowledges support by a scholarship of the Alexander S. Onassis Public Benefit Foundation.}\thanks{sleonardos@math.uoa.gr}}
\author[1]{Costis Melolidakis\thanks{cmelol@math.uoa.gr}}
\affil[1]{\normalsize Department of Mathematics, National and Kapodistrian University of Athens, Panepistimioupolis GR - 157 84, Athens, Greece}

\begin{document}
\maketitle
\begin{abstract}
Given a bimatrix game, the associated leadership or commitment games are defined as the games at which one player, the leader, commits to a (possibly mixed) strategy and the other player, the follower, chooses his strategy after having observed the irrevocable commitment of the leader. Based on a result by \cite{St10}, the notions of commitment value  and  commitment optimal strategies for each player are discussed as a possible solution concept. It is shown that in non-degenerate bimatrix games (a) pure commitment optimal strategies together with the follower's best response constitute Nash equilibria, and (b) strategies that participate in a completely mixed Nash equilibrium are strictly worse than commitment optimal strategies, provided they are not matrix game optimal. For various classes of bimatrix games that generalize zero sum games, the relationship between the maximin value of the leader's payoff matrix, the Nash equilibrium payoff and the commitment optimal value is discussed. For the Traveler's Dilemma, the commitment optimal strategy and commitment value for the leader are evaluated and seem more acceptable as a solution than the unique Nash equilibrium. Finally, the relationship between commitment optimal strategies and Nash equilibria in $2 \times 2$ bimatrix games is thoroughly examined and in addition, necessary and sufficient conditions for the follower to be worse off at the equilibrium of the leadership game than at any Nash equilibrium of the simultaneous move game are provided.
\end{abstract}
\noindent \\[0.2cm] \textbf{Keywords:} Bimatrix Game, Nash Equilibrium, Subgame Perfect, Commitment Optimal, Commitment Value, Weakly Unilaterally Competitive Games,  Pure Strategy Equilibrium, \\[0.2cm]
\textbf{JEL Classification:} C72, \textbf{AMS 2010 Subject Classification} Primary: 91A05, Secondary: 91A10, 91A40.

\section{Introduction}\label{intro}
In the 1920's, when trying to formalize zero-sum games and propose a solution concept for them, both E. Borel and J. von Neumann approached two-person zero-sum games from an \emph{optimization} point of view. As payoff functions depend on both players actions, direct optimization of a player's payoff does not make sense and both Borel and von Neumann reached the concept of the players' \emph{security} or \emph{safety level}. It was defined as the best among worst possible outcomes for that player and was to be taken as the ``value of the game'' for each player\footnote{For a discussion of the details (sometimes controversial) of their contribution the reader is referred to \cite{Dim96}, \cite{Dim-Mec10}, and \cite{Kj01}.}. This formulation led to the famous minimax theorem, which, stating that the values of both players coincide, established the common value as the indisputable solution concept for these games.\par
To present this concept in their classic book, J. von Neumann and O. Morgenstern proposed two auxiliary games: The ``minorant'' game $\Gamma_1$, in which player I chooses his mixed strategy $x$ first, and then II, in full knowledge of $x$ (but not of its realization), chooses his mixed strategy $y$, and the ``majorant'' game $\Gamma_2$, in which the order of the players' moves is reversed. This scheme was proposed so that the optimization of each player's utility would make sense: 
\begin{quote}
\emph{The introduction of these two games $\Gamma_1$, $\Gamma_2$ achieves this: It ought to be evident by common sense -- and we shall also establish it by an exact discussion -- that for $\Gamma_1$, $\Gamma_2$ the ``best way of playing''-- i.e. the concept of rational behavior -- has a clear meaning\footnote{For more on this subject, see chapters 14 and 17 of \cite{Mo53} -- the excerpt is from p.100.}.}
\end{quote}
Hence, to be able to use individual rationality (i.e. the maximization of a player's utility) in deriving a solution, a leader-follower scheme was utilized and the notion of a common safety (or security) level as the solution evolved naturally from optimality considerations. \par

However, generalizing this approach to two person, non zero-sum, non-cooperative games came to a dead end. The reason is that implicit to the Borel-von Neumann approach are three different points of view which do not necessarily agree on (non zero-sum) bimatrix games. The first point of view is that of the players optimizing against the worst possible outcome. For zero-sum games, individual rationality of the opponent coincides with assuming that he is out there to destroy us, but this assumption seems unreasonable for non-zero sum games. The second point of view is that of the players optimizing against the mixed strategy of each other simultaneously, where both mixed strategies are considered to be known. Although this seemed difficult to accept from the optimality point of view, it was this approach that later led to the widely accepted Nash equilibrium concept, which addressed the problem from an \emph{equilibrium} rather than an optimization perspective. The third point of view is that of the players optimizing in a leader-follower sequence, i.e. the follower optimizing against a known \emph{mixed} strategy of the leader, and thus obtain a ``value of the game'' for the leader. Applied to bimatrix games and assuming that irrevocable commitment on mixed strategies by the players is possible, the objections one may raise to the third point of view are firstly, what factors determine the order of play for the players and secondly, how is the leader going to optimize if the follower has a non-unique best response? \par

The problems that emerge in the leader-follower approach as well as the comparison of simultaneous and sequential move versions of the same underlying game, the latter under various assumptions on the sense of commitment of the leader, have been a topic of continuous research in the game theory literature. What these comparisons actually do is more or less investigating the relationship between the three ``points of view'' we talked about, often for specific classes of games. Since this literature has a wide variety of themes, we will refer only to papers related to the particular questions we examine in the present paper. All such references are cited in detail and in context to each particular question raised in the main text. However, even if it is not our topic, a small discussion of the rationale behind the leader-follower approach in bimatrix games and of the two objections mentioned above should be helpful in appreciating the rest of this paper.\par

\Citet{Ro91} examines normal form bimatrix games and related sequential versions, where one of the players commits to any of his (mixed) strategies and this commitment becomes known to the other player prior to his strategy selection. Rosenthal defines \textit{commitment-robust} equilibria to be Nash equilibria of the simultaneous move game that are subgame perfect equilibria of both the related sequential move versions and argues that Nash equilibria that fail this property ought to be questionable, if there is sufficient flexibility in the rules of the game.\par 

In an attempt to deal with the choice of the leader/follower, i.e. the first objection noted above, \cite{Ha93} propose a two stage generated model for $2 \times 2$ bimatrix games that determines ``endogenously'' whether the game will be played sequentially and, in that case, the ordering of the two players. Using this model for non-degenerate bimatrix games, \cite{Hu96} address a question related to that raised by \cite{Ro91}: \emph{when is an equilibrium of the original game an equilibrium of the generated game also?} Such equilibria are termed ``viable''. They prove that an equilibrium of the original game in \emph{mixed} strategies is viable if and only if no player has an incentive to move first in the ordinary, sequential move, commitment game. It is clear that the Hamilton-Slutsky two stage generated model was chosen by van Damme and Hurkens as a tool to deal with the ordering of the players. However, although the imposition of a super-game over the original one, with its own rules and assumptions, leads to interesting results, it is questionable whether it actually resolves the issue and should be preferable to the original setup. It is our opinion that an answer that would settle this objection conclusively has yet to be proposed.\par

The second objection raised when one approaches bimatrix games by a leader-follower approach, namely that of the leader's choice when the follower has more than one best responses to the leader's strategy, has been settled by \citet{St10}. Prior to that publication, the prevailing approach was to consider either the case where the follower chooses his best response so as to accommodate the leader or (on the contrary) he chooses his best response so as to harm the leader\footnote{For a detailed exposition of the theory developed in the latter case, see Chapter 3.6,``The Stackelberg Equilibrium Solution'', of \citet{Bol99}.}. It is questionable whether both these approaches agree with the principle of individual rationality of the follower and, anyway, they leave the question of real play for the leader open: how could his play depend on the \emph{assumption} that the follower would play in this or that fashion? Certainly, the question pertains to degenerate bimatrix games, but yet, ignoring degeneracy is not a satisfactory answer. \Citet{St10} show that when a bimatrix game is played sequentially with the mixed strategy of the leader being observed by the follower before he makes his move, then all subgame perfect equilibria payoffs for the leader form a closed interval of the form $[\alpha^L, \alpha^H]$, where $\alpha^L$ is a payoff that the leader \emph{can guarantee} (i.e. induce the follower to give a best response that results to $\alpha^L$). This interval collapses to a point when the bimatrix game is non-degenerate.\par

An important implication of their result is that, for bimatrix games, $\alpha^L$ is precisely the safety level of the leader, i.e. the ``value'' of the minorant game or majorant game (according to who plays first) in the \citet{Mo53} sense. However, these two values will generally not be the same and of course, a player's payoff will be different on the inducible subgame perfect equilibrium if he plays first or second. We shall call the safety level of the leader in the minorant game  (resp. majorant) \emph{commitment value for player I (resp. player II)} and any strategy of the leader that guarantees him his commitment value \emph{commitment optimal}. \par

Based on these observations, there are three notions that their relationship needs to be studied further: matrix values and corresponding max-min strategies in the individual matrix games, Nash equilibria and their payoffs in the simultaneous move bimatrix game and finally, commitment values and commitment optimal strategies in the minorant/majorant games. \par

In the present paper, we try to investigate these relationships. Our results are far from being complete, however we got some interesting characterizations. Firstly, we derive two general properties of Nash equilibria and commitment optimal strategies. We show that when the leader has a pure commitment optimal strategy in a non-degenerate leader-follower game, then this strategy together with the follower's best response form a Nash equilibrium in pure strategies in the underlying bimatrix game. Hence, if it is possible to improve upon his Nash equilibria payoffs, the leader must use mixed strategies in the leader-follower non-degenerate game, which underlines the importance of mixed strategies per se. This result is not true in the case of degeneracy. Also, we show that in a non-degenerate bimatrix game a player can strictly improve his payoff at a completely mixed Nash equilibrium by commitment, provided his Nash equilibrium strategy is not matrix game optimal (i.e. maximin). \par
Secondly, taking player I to be the leader (without loss of generality) and letting $v_A$ denote the value of the leader's matrix game, we discuss the validity of the equation $v_A=\alpha^L=\alpha^H$ for various classes of bimatrix games that are viewed as generalizations of zero-sum games.\footnote{Observe that if this equation is valid, then all Nash equilibria payoffs for the leader will also be equal to $v_A$ since, as \citet{St10} show, the payoff for the leader at any Nash equilibrium of the bimatrix game is less or equal to $\alpha^H$ .} Of course, this equation is true for zero sum games. We show that the equation obtains for weakly unilaterally competitive games, but is false for other generalizations of zero-sum games, such us almost strictly competitive games, pre-tight, best response equivalent to zero sum games, and strategically zero sum games. Of particular interest is the case where the payoff at all Nash equilibria of a bimatrix game is equal to the value of the leader's matrix game $v_A$ but less than his commitment value. The Nash equilibrium is a questionable solution concept for such games. \par
Thirdly, we discuss the game known as the Traveler's Dilemma-TrD (see the pertinent section for references on TrD), which is a typical example of a bimatrix game whose unique Nash equilibrium is compelling under standard equilibrium arguments (e.g. domination of strategies) but  unattractive under both optimality considerations and common sense, and also unsupported by experimental game theory\footnote{Concerning experiments on TrD, one is referred to \citet{Ba11}. On an experiment in which the participants were members of the Game Theory Society, see \citet{Be05}.}. For the TrD, we derive the commitment value and the commitment optimal strategy of the players (the same for both due to symmetry), which are very close both to the Pareto optimal outcome and to the behavior of game participants in experiments. It is noteworthy that the follower's payoff is then identical to that of the leader and thus, the symmetry of the game is preserved by the leader-follower solution. TrD is an example of a game where the payoff at Nash equilibria equals to the matrix game value of the leader which is strictly less than his commitment value. \par
Finally, for the case of $2 \times 2$ bimatrix games, we exhaustively examine the relationship between commitment optimal strategies of the leader, optimal responses of the follower and Nash equilibrium strategies of both. We also provide necessary and sufficient conditions for the follower to be worse off at the equilibrium of the leadership game than at any Nash equilibrium of the simultaneous move game.

\subsection{Outline}
The rest of the paper is structured as follows. In section \ref{definitions}, formal definitions and notations for bimatrix games and their associated leadership games are given. In addition, some results from \cite{St10} are presented, which are used subsequently. Section \ref{properties} presents properties and compares matrix game values, Nash equilibria payoffs and commitment values. In section \ref{classes}, special classes of bimatrix games that have been proposed as generalizations of zero-sum games are examined with respect to the relationship between maximin, optimal commitment and Nash equilibria strategies. In section \ref{trd}, we discuss the traveler's dilemma. Finally, section \ref{twobim} is devoted to $2\times 2$ bimatrix games, discussing the relationship between commitment optimal strategies and Nash equilibrium strategies, and also comparing the payoffs of the follower in these two cases.

\section{Definitions}\label{definitions}
We consider the mixed extension of an $m \times n$ bimatrix game $(A,B)$ played by players I and II, which we denote shortly by $\Gamma$. The sets of pure strategies of the players are $I=\{1,2,\ldots, m\}, J=\{m+1,m+2,\ldots,m+n\}$ respectively. A pure strategy of player I will be denoted by $s_i$ or simply by $i$, for $i=1,2,\ldots,m$, when confusion may not arise. Similarly, a pure strategy of player II will be denoted by $t_j$ or simply by $j$, for $j=m+1,m+2,\ldots,m+n$. The sets of mixed strategies of player I (resp. player II) will be denoted by $X$ (resp. $Y$), where $X= \Delta (m)$ and $Y=\Delta (n)$ are the $m-1$ dimensional and $n-1$ dimensional probability simplexes. For $x\in X$ and $y \in Y$, the payoffs of player I and II are given by $\alpha\(x,y\):=x^TAy$ and $\beta\(x,y\):=x^TBy$. The value of the matrix game of player I will be denoted by $v_A$. It could be considered as I's safety level, since I may guarantee the value of the matrix game $A$ no matter what strategy player II chooses; however we will avoid this interpretation in view of our previous discussion on commitment value. Of course, $v_A = \max_{x\in X}\min_{j\in J}\alpha\(x,j\)$. The corresponding quantity for player II is $v_B = \max_{y\in Y}\min_{i \in I}\beta\(i,y\)$, the value of $B^T$. Given a strategy $x \in X$ of player I, a strategy $y$ of player II is a best response to $x$ if $\beta\(x,y\)=\max_{y'\in Y}\beta\(x,y'\)$. In symbols, $\brii\(x\):=\{y\in Y: \beta\(x,y\)\ge \beta\(x,y'\), \forall y' \in Y\}$. Note, that $\brii\(x\)=\conv\left\{j\in J: j\in \brii\(x\)\right\}$. A strategy $x\in X$ is strongly dominated by a strategy $x'\in X$ if $\alpha\(x',y\)>\alpha\(x,y\)$ for all $y \in Y$ and weakly dominated by a strategy $x''\in X$ if $\alpha\(x'',y\)\ge \alpha\(x,y\)$ for all $y\in Y$, with strict inequality for at least one $y\in Y$.\par

For any $j\in J$, the best reply region of $j$ is the set $X\(j\)\subseteq X$ on which $j$ is a best reply, i.e. $X\(j\):=\left\{x\in X: j\in \brii\(x\)\right\}$. The sets $Y\(i\)$ for $i\in I$ are defined similarly. The edges of $X\(j\)$ are denoted by $C\(j\)$. By standard convexity arguments, $\argmax_{x\in X\(j\)}\alpha\(x,j\)\cap {C\(j\)} \neq \emptyset$ in any bimatrix game. We say that $X\(j\)$ is full-dimensional, if there is $x\in X\(j\)$ such that $x_i>0$ for all $i=1,2,\ldots,m$, i.e. if the interior $X^{\text{o}}\(j\)$ is not empty. Let $D:=\{j\in J: X^{\text{o}}\(j\)\neq \emptyset\}$. \Citet{St04} show that $j \in D$ if and only if $j$ is not (weakly) dominated. For $j\in J$, let $\mathcal E\(j\)=\{k\in J: \beta\(\cdot,j\)\equiv \beta\(\cdot,k\)\}$, i.e. $\mathcal E\(j\)$ denotes the set of all pure strategies $k \in J, k\neq j$, that are payoff equivalent to $j$.\par 

Given a bimatrix game $\Gamma$, we consider two associated \emph{leadership (leader-follower)} or \emph{commitment games}, denoted by $\gi$ and $\gii$. We define $\gi$ to be the game at which player I is the leader, i.e. he moves first and commits to a strategy, possibly mixed, and player II is the follower, i.e. he moves second, after having observed the strategy choice of player I. Formally, the strategy set of player I is $X$ as in $\Gamma$, while the strategy set of player II is the set of measurable functions $f:X\mapsto Y$. The payoffs of the players in $\gi$ are determined by their payoff functions in $\Gamma$, that is $\alpha\(x,f\(x\)\)=x^TAf\(x\)$ and $\beta\(x,f\(x\)\)=x^TBf\(x\)$. When considering best replies of player II, we will usually restrict attention to pure strategies since there is no need for the follower to employ a mixed strategy. Similarly, we define $\gii$ as the game where player II moves first and player I moves second. In what follows, we use $\gi$ with player I as the leader in a generic fashion, but, unless stated otherwise, results apply also to $\gii$.\par

A bimatrix game is non-degenerate if no mixed strategy of any player has more pure best replies than the size of its support. When considering the game $\gi$ we will need this property to hold only for mixed strategies of the player that moves first, i.e. player I. This motivates the following definition
\begin{definition}\label{nond}
A bimatrix game $\Gamma$ is non-degenerate for player $i$, for $i=\{\mathrm{I, II}\}$, if no mixed strategy of player $i$ has more pure best replies among the strategies of player $j$ than the size of its support.
\end{definition}
For the leadership games $\gi$ (and $\gii$), we use the subgame perfect equilibrium as a solution concept, under which an optimal strategy $f\(\cdot\)$ of the follower is a best reply to any $x\in X$. For the simultaneous move game $\Gamma$, we will be interested in Nash equilibria and, in some cases, in correlated equilibria (\cite{Au74}) or in coarse correlated equilibria (\cite{Mo78}). \par

A Nash equilibrium strategy profile will be denoted by $\(x^N,y^N\)$, with $x^N\in X$ and $y^N\in Y$. The set of all Nash equilibria strategy profiles of the bimatrix game $\Gamma$ will be denoted by $NE\(\Gamma\)$ and the set of all Nash equilibria payoffs by $NE\Pi\(\Gamma\)$. The payoffs of player I and II at a Nash equilibrium will be denoted by $\alpha^N:=\alpha\(x^N,y^N\)$ and $\beta^N:=\beta\(x^N, y^N\)$ respectively. Finally, we write $NE\(X\), \(NE\(Y\)\)$ for the set of all strategies of player I (resp. II) that participate in some Nash equilibrium of the simultaneous move game.\par 

\subsection{Existing results: equilibrium payoffs of the leader}\label{exist}
The present work builds upon results that appeared recently in the literature, some of which we present here. \Citet{St10} prove that in a degenerate bimatrix game, the subgame perfect equilibria payoffs of the leader form an interval $\left[\alpha^L, \alpha^H\right]$. The lowest leader equilibrium payoff $\alpha^L$ is given by the expression 
\begin{equation}\label{alphaL}\alpha^L=\max_{j\in D}\max_{x\in X\(j\)}\min_{k\in \mathcal E\(j\)}\alpha\(x,k\)\end{equation}
and the highest leader equilibrium payoff $\alpha^H$ is given by 
\begin{equation}\label{alphaH}\alpha^H=\max_{x\in X}\max_{j\in \brii\(x\)}\alpha\(x,j\)=\max_{j \in J}\max_{x\in X\(j\)}\alpha\(x,j\) \end{equation}
If the game is non-degenerate (or non-degenerate for the leader), then the leader has a unique subgame perfect equilibrium payoff in $\gi$. In this case, the expressions in (\ref{alphaL}) and (\ref{alphaH}) coincide. In fact, less than non-degeneracy is required for the equality (hence uniqueness) to hold. If any best reply region is full-dimensional or empty and if there are no payoff equivalent strategies, i.e. if $\mathcal E\(j\)=\left\{j\right\}$ for any $j\in J$, then expression \eqref{alphaL} yields the same value as expression \eqref{alphaH}, i.e. $\alpha^L=\max_{j\in J}\max_{x\in X\(j\)}\alpha\(x,j\)=\alpha^H$. See also Example \ref{-321-303} in section \ref{cpure} for such a case. \par

For the non-degenerate case, the intuition behind deriving the unique leader equilibrium payoff $\alpha^L=\alpha^H$ is the following. The leader commits to a strategy $x\in X$ and the follower gives a best response, which the leader may force to be the best possible for him among player II's best responses. This follows from an $\epsilon$-argument. As a sketch, for any strategy $x \in X$ that admits more that one pure best replies, player I may sacrifice an $\epsilon>0$ and move to a nearby strategy that admits a unique pure best reply. \Citet{St04} call such a strategy of the follower \textit{inducible}, in the sense that by sacrificing $\epsilon>0$ the leader may induce the follower to use it. By the non-degeneracy property, this can be done for \textit{any} pure best reply of the follower against $x$, for any $x\in X$. In equilibrium there is no sacrifice and the result obtains. \par

In non-degenerate bimatrix games, the unique leader equilibrium payoff is identified as the highest payoff for that player on an edge of his part of the Lemke-Howson diagram (meaning the part of the diagram on this player's simplex).\par

In games that are degenerate or at least degenerate for the leader (cf. definition \ref{nond}), the reason for the possible difference between the lowest ($\alpha^L$) and highest ($\alpha^H$) leader equilibrium payoff is that not all pure best replies of the follower may be inducible. This is so for the case of weakly dominated strategies, which have not full-dimensional best reply regions and for the case of payoff equivalent strategies which have best reply regions that fully coincide. Again, $\alpha^H$ corresponds to the best possible subgame equilibrium payoff for the leader. To obtain the lowest leader payoff $\alpha^L$, \cite{St04} reduce the original game to a game that is non-degenerate for the leader, by ignoring his payoffs against weakly dominated strategies and solving for his safety level against payoff equivalent strategies of the follower. The unique leader equilibrium payoff in this reduced, non-degenerate (for the leader) game, is now the lowest payoff $\alpha^L$ that he may guarantee in the original leadership game.\par

If $l$ denotes the lowest and $h$ the highest Nash equilibrium payoff of player I in $\Gamma$, \citet{St10} show that $l\le\alpha^L$ and\footnote{The last inequality holds also for any correlated equilibrium payoff, i.e. the highest leader payoff of a player is at least as high as his highest correlated equilibrium payoff in the simultaneous move game.}  $h\le \alpha^H$. Together with the trivial inequality $v_A\le l\le h$, their result establishes lower and upper bounds for the Nash equilibria payoffs of a player. So, in degenerate games with 
$\alpha^L<\alpha^H$
\begin{align}\begin{split}\label{bounddeg} v_A&\le l \le \alpha^L \\[0.2cm] v_A&\le l \le h \le \alpha^H\end{split}\end{align}
while for non-degenerate games, since $\alpha^L=\alpha^H$, (\ref{bounddeg}) simplifies to
\begin{equation}\label{boundndeg}v_A\le l\le h\le \alpha^L\end{equation}\par

We have already noted that $\alpha^L$ should not be viewed as the lowest subgame perfect equilibrium payoff for the leader, but rather as his commitment value in the leadership game, i.e. a payoff that he may guarantee assuming the other player is rational (utility maximizer). 

\subsection{Motivation: Safety levels and Nash equilibria payoffs}
Hence, applied to bimatrix games, the optimization point of view leads to $\alpha^L$ as a candidate for the generalization of the notion of safety level. As we saw, \citet{Mo53} take the leader-follower approach to discuss the safety level $v_A$ in matrix games, for which of course $v_A=\alpha^L$. At first sight, the payoff $\alpha^L$ refers to a different game $\gi$ than the payoffs $\alpha^N, v_A$ with which we want to compare it. However, as \citet{Co16} argues, the leadership equilibrium should be viewed as a distinct solution concept for the game $\Gamma$ itself and not as an application of the Nash equilibrium concept on the different game $\gi$. So, one wonders whether a theory can be developed for the solution of non-zero sum games which will not originate from the equilibrium point of view but will be based on the leader-follower approach, which is generated from the optimization point of view. In that case, $\alpha^L$ will play a central role for the case of bimatrix games.\par
 
Based on our discussion thus far and on equations \eqref{bounddeg} and \eqref{boundndeg}, one is motivated to raise certain questions. Firstly, is it possible to characterize all bimatrix games for which $h\le \alpha^L$? This inequality, which already holds for non-degenerate games, leads to discarding all Nash equilibria if the rules of the game permit commitment, as is often the case. Certainly, in that case a new problem appears: How will the leader be chosen among the two players? Secondly, for which classes of bimatrix games does equation \eqref{boundndeg} collapse to $v_A=\alpha^L$, i.e. the leader may not guarantee more than the safety level of his payoff matrix? 

\section{Nash Equilibria, maximin strategies and commitment optimal stra\-tegies}\label{properties}
Along with the leader payoffs at subgame perfect equilibria of $\gi$, we want to study his equilibrium strategies. In agreement with the notation $\alpha^L$ and $\alpha^H$, we denote by $x^L$ and $x^H$ the strategies that the leader uses to attain these payoffs. However, $x^L$ and $x^H$ may be non-unique.
\begin{example} In the $3\times3$ bimatrix game $\Gamma$ with payoff matrices
\[A=\begin{pmatrix*}[r] 4 & 1 & 0\\ 3  & 2 & 0 \\ 0 & 0 & 3.5\end{pmatrix*}\qquad B=\begin{pmatrix*}[r] 1 & 2 & 0\\ 4 & 3 & 0 \\ 0 & 0 & 1\end{pmatrix*}\]
$\alpha^L=\alpha^H=3.5$, but this payoff can be achieved by player I (the leader) in $\gi$ with two different commitment strategies: either $x^L_1=\(\frac12,\frac12,0\)$, which induces the follower to play $t_4$ or $x^L_2=\(0,0,1\)$ which induces the follower to play $t_6$. This is despite the fact that the game is non-degenerate. 
\end{example}
To proceed, we give a formal definition of commitment value and commitment optimal strategies.
\begin{definition}\label{comm}
Let $\Gamma$ be a bimatrix game and let $\gi$ and $\gii$ be the associated leadership (leader-follower) or commitment games, where \I is the leader in $\gi$ and \II is the leader in $\gii$. Then, the leader's inducible lower subgame perfect equilibrium payoff $\alpha^L$ in $\gi$ (resp. in $\gii$) will be called commitment value for \I (resp. $\mathrm{II}$). A strategy that guarantees his commitment value to a player will be called commitment optimal.
\end{definition}
As the example above shows, the set of commitment optimal strategies for a player need not be a singleton. We will take player I as the ``default'' leader player and we will denote by $X^L$ the set of his \emph{commitment optimal} strategies, with generic element $x^L \in X^L$. A strategy that the leader may induce the follower to use when playing commitment optimally\footnote{For a given $x^L$, $j^F(x^L)$ is actually not unique. However, on the set of those $j^F$s the follower's payoff is constant, so we don't really care about the non-uniqueness of $j^F(x^L)$.} by using $x^L\in X^L$ will be denoted by $j^F(x^L)$ and the corresponding payoff of the follower will be denoted by $\beta^F(x^L)$ or simply $\beta^F$, i.e. $\beta^F:=\beta\(x^L,j^F(x^L)\)$, where $j^F(x^L) \in J$ such that $\alpha\(x^L, j^F(x^L)\)=\alpha^L$. Notice that at different strategy pairs $\(x^L, j^F(x^L)\)$ the payoff of the leader is constant and equal to $\alpha^L$ but the follower may obtain different payoffs $\beta^F$.
 
\subsection{Monotonicity of the bounds: matrix game value, commitment value and equilibria payoffs}
We start with the observation that the lower and upper bounds in relations \eqref{bounddeg} and \eqref{boundndeg}, i.e. $v_A$ and $\alpha^L$ (resp. $\alpha^H$) exhibit certain monotonicity relations to the sizes of the pure strategy spaces (i.e. the number of pure strategies) of the players. The proof is immediate from the definitions and thus ommitted.
\begin{lemma}\label{mon} Let $|I|=m$ and $|J|=n$ be the numbers of pure strategies of player \I and \II respectively in a bimatrix game $\Gamma$. Then 
\begin{enumerate}[leftmargin=0cm,itemindent=.5cm, labelwidth=\itemindent,labelsep=0cm, noitemsep, align=left]
\item The value of the matrix game of player \I in $\Gamma$, $v_A$, is a non-decreasing function of $m$ and a non-increasing function of $n$.  
\item The lowest and highest leader payoffs $\alpha^L$ and $\alpha^H$ of player \I in $\gi$ are non-decreasing functions of $m$, but not necessarily non-increasing in $n$.
\item The Nash equilibria payoffs of player \I do not have a certain monotonicity relation to the number of the player's own strategies or to the number of the other player's strategies.
\end{enumerate}
\end{lemma}
Although obvious, Lemma \ref{mon} highlights an undesired property of the Nash equilibrium. Given the strategies of his opponent, having more choices should offer a strategic advantage to a player. While this is indeed the case in terms of the leader's commitment value $\alpha^L$, his highest subgame perfect equilibrium payoff $\alpha^H$, and his matrix game value $v_A$, having more options may well be harmful in terms of his Nash equilibria payoffs in $\Gamma$. \par 

It is easy to construct such an example by referring to the TrD (see section \ref{trd}). There, consider first the game having the same payoff functions $\alpha, \beta$, but strategy spaces $I=\{100\}, J=\{2,3,\ldots,100\}$ and start increasing the strategy space of player I by adding 99, then 98, etc.\par 

The introduction/removal of strongly dominated strategies affects these bounds in a non-trivial way. For example, in evaluating the matrix game value $v_A$, the $\min$ is taken against all strategies in $J$, including strongly dominated strategies of player II. In view of Lemma \ref{mon}, this means that $v_A$ may decrease by the addition of a strongly dominated strategy to player II's strategies.\par

On the other hand, Nash equilibria payoffs and leader payoffs remain unaffected if strongly dominated strategies of the other player are introduced/removed. A more interesting case occurs when I has strategies that are strongly dominated in $\Gamma$ but not in $\gi$. Such strategies do not affect $v_A$ or Nash equilibrium payoffs $\alpha^N$, but may improve his leader payoff bounds $[\alpha^L, \alpha^H] $. This subject is addressed in detail later, see example \ref{31019} and section \ref{trd} for instances of this case.

\subsection{Nash equilibria and pure commitment optimal strategies}\label{cpure}
We show that for non-degenerate bimatrix games any \textit{pure} commitment optimal strategy of I, together with II's best response to it, constitute a Nash equilibrium of $\Gamma$. In the proof we make use of an observation by \citet{St10}, namely that in a non-degenerate game any best reply region $X\(j\)$ is either empty or full-dimensional, i.e. either $X\(j\)=\emptyset$ or $j \in D$.\par

\begin{proposition}\label{pure}
If the bimatrix game $\Gamma$ is non-degenerate for player $\mathrm I$, and if \I has a pure commitment optimal strategy $i^L$, then the strategy profile $\(i^L, j^F\)$ is a pure strategy Nash equilibrium of $\Gamma$.    
\end{proposition}
\begin{proof}
Let $i^L \in X^L$ and let $j^F \in \brii\(i^L\)$. Since the game is non-degenerate for I, $j^F$ is unique, i.e  $i^L \notin X\(j\)$ for any other $j \in J, j \neq j^F$. Since $i^L \in X\(j^F\)$, $X\(j^F\)$ must be full-dimensional, so that for any $i\in I$ and for $\epsilon>0$ sufficiently small, the mixed strategy $x^{L}_{\epsilon,i}:=\(1-\epsilon\)\(i^L\)+\epsilon \(i\)$ lies only in $X\(j^F\)$. Since $\alpha\(i^L, j^F\)=\max_{x\in X\(j^F\)} \alpha\(x, j^F\)$ and $\alpha\(\cdot, j^F\)$ is linear, we conclude that $\alpha\(i^L,j^F\) \ge \alpha\(x^L_{\epsilon,i}, j^F\)=\(1-\epsilon\)\alpha\(i^L,j^F\)+\epsilon\alpha\(i,j^F\) \forall i \in I$. But then, $\alpha\(i^L,j^F\) \ge \alpha\(i,j^F\) \forall i \in I$, i.e. $i^L \in \bri\(j^F\)$. Hence, $\(i^L, j^F\)$ are both best responses one to the other in $\Gamma$.
\end{proof}
The construction of the mixed strategies $x^L_{\epsilon,i}$, $i\in I$, has been employed by \cite{St04} and \cite{St10} in various proofs. If the game is degenerate, then the pure strategy $i^L$ may also belong to another $X\(j\)$ for some $j \in J, j\neq j^F$, and the statement of Proposition \ref{pure} is not always true. This is highlighted by the following example. 
\begin{example}\label{-321-303}
The $3\times 2$ bimatrix game $\Gamma$ with pure strategy spaces $I=\{1,2,3\}, J=\{4,5\}$ and payoff matrices $\(A, B\)$ 
\[A=\begin{pmatrix*}[r]-3 & 2\\ 1 & -3 \\ 0 & 3\end{pmatrix*}\qquad B=\begin{pmatrix*}[r]0 & 0\\ 1 & 2 \\ 3 & -1\end{pmatrix*}\]
has a unique Nash equilibrium $\(x^N,y^N\)=\(\(0, \frac45, \frac15\),\(\frac67,\frac17\)\)$ with payoffs $\(\alpha^N,\beta^N\)=\(\frac37,\frac75\)$. The pure strategy $s_1$ of I has two pure best replies, i.e. $\brii\(s_1\)=\conv\{t_4,t_5\}$, hence the game is degenerate (for player I). Nevertheless, both best reply regions $X\(4\)$ and $X\(5\)$ have full dimension and therefore player I's equilibrium payoff in the leadership game $\gi$ is unique (c.f. section \ref{exist}) and may be determined by relation \eqref{alphaH}. The edges of the best reply regions and the corresponding payoffs for player I are
\[\begin{array}{ll}
C\(4\)=\left\{s_1,\(0, 0.8, 0.2\),s_3\right\}, \text{with }\, \alpha\(x,j^F=4\)_{x\in C\(4\)}=\left\{-3, 0.8, 0\right\}\\
C\(5\)=\left\{s_1,\(0, 0.8, 0.2\),s_2\right\}, \text{with }\, \alpha\(x,j^F=5\)_{x\in C\(5\)}=\left\{2, -1.8, -3\right\} \end{array}\]
giving that $\alpha^L=2$, with $\(x^L,j^F\)=\(s_1,t_5\)$ which, however, is not a pure strategy equilibrium of $\(A,B\)$. 
\end{example} 
In view of Proposition \ref{pure}, a necessary condition for the commitment value to be strictly better than all Nash equilibria payoffs in a non-degenerate bimatrix game, i.e. for  $\alpha^L>h$, is that all strategies in $X^L$ are mixed. In other words, actual mixed strategies have to be used if the leader is to improve his payoff over all Nash equilibria. This is in line with the ``concealment'' interpretation of mixed strategies as expressed in \citet{Mo53} and is to be expected since the leader-follower game is precisely the sort of game that von Neumann and O. Morgenstern were considering when arguing about the need of concealment\footnote{The need for actually mixing pure strategies originates from the optimization point of view, while the Bayesian interpretation of mixing as uncertainty on a player's type by the other players originates from the equilibrium point of view. More on the interpretation of mixed strategies can be found in \cite{Re04}.}.

\subsection{Completely mixed Nash equilibrium strategies vs commitment optimal strategies}

The next property states that a player's strategy at a Nash equilibrium is strictly dominated by his commitment optimal strategy provided (a) the bimatrix game is non-degenerate, (b) the Nash equilibrium under consideration is completely mixed, and (c) the equilibrium strategy under consideration is not matrix game optimal (i.e. maximin).\par
 
A Nash equilibrium $(x^N,y^N)$ is completely mixed if all pure strategies of both players are played with positive probability, which implies that any $i\in I$ is a best response against $y^N$ and similarly any $j\in J$ is a best response against $x^N$. Moreover, the non-degeneracy property implies that there is no other completely mixed equilibrium and that the supports of $x^N, y^N$ have equal size, i.e. $|I|=|J|$.
\begin{proposition} \label{mixed} Let $\Gamma$ be a non-degenerate bimatrix game with a completely mixed Nash equilibrium $(x^N, y^N)$, such that $x^N$ is not a matrix game optimal (maximin) strategy. Then $\alpha^L>\alpha^N$.
\end{proposition}
\begin{proof} By Lemma 1 of \cite{Pr11}, $x^N$ cannot be a column equalizer in the payoff matrix $A$ of player I, since otherwise $x^N$ would be a maximin strategy. Hence, there exist $j_1,j_2$ in $J$ such that $\alpha(x^N,j_1)>\alpha^N>\alpha(x^N,j_2)$. Since, $j_1$ is inducible by the non-degeneracy property, the point $\alpha(x^N,j_1)$ constitutes a payoff that player I can guarantee for himself in $\gi$ and hence $\alpha^L\ge \alpha(x^N,j_1)$. 
\end{proof}
Example \ref{4132} describes a game where Proposition \ref{mixed} holds for both players. If $x^N$ is maximin then an improvement may not be possible, as in the $2\times 2$ matrix game of matching pennies. In case Proposition \ref{mixed} applies, the strategy profile $(x^N,j_1)$  Pareto-dominates the completely mixed Nash equilibrium, since by construction, $\alpha(x^N,j_1)>a^N$ and $\beta(x^N,j_1)=b^N$. 

\section{Commitment value and Nash equilibria payoffs in generalizations of matrix games}\label{classes}
Many classes of games that extend two-person zero-sum games have been studied in the literature. Among others, one is referred to \cite{Au61, Ka92}, \cite{Mo78} and \cite{Be02}, who generalize zero-sum games in different ways (see the review by \cite{Vi06}). A natural question is whether all three solution concepts (i.e. matrix game value, Nash equilibria payoffs, and commitment value) coincide on such generalizations, as is the case for zero-sum games, i.e whether
\begin{equation}\label{three} v_A=\alpha^L=\alpha^H .\end{equation}  
Of interest is also the case where all Nash equilibria payoffs of the leader are equal to his matrix game value, but his commitment value is strictly higher, i.e. 
\begin{equation}\label{two} v_A=h < \alpha^L .\end{equation}

In section \ref{wucg}, it is shown that \eqref{three} is true for the class of weakly unilaterally competitive (wuc) games, firstly defined by \cite{Ka92}. The wuc games strictly include the classes of zero-sum and strictly competitive games. Recently, interest in wuc games was reignited in view of the sufficient conditions for the existence of pure strategy equilibria in such games that were given by \cite{Ii16}. \par
Equation \eqref{three} is also valid in the class of a-cooperative games. However, for other generalizations of zero-sum games, namely pre-tight, best response equivalent to zero-sum, and almost strictly competitive games, \eqref{three} is not valid. Equation $v_A=h$ holds in some of these classes and then \eqref{two} may be true in certain cases (for details see section \ref{other}). \par

Classes generalizing zero sum games which satisfy \eqref{three}, such as wuc or a-cooperative games, retain the flavor of pure antagonism that characterizes zero-sum games. All solution concepts on them coincide and no controversies on what constitutes an optimal behavior for the players may arise. On the other hand, we expect Nash equilibrium to be a questionable solution concept for games in classes generalizing zero sum games which satisfy \eqref{two}. Such ``bad behavior'' cases may be found in the class of pre-tight games, a typical example being the TrD. For a detailed exposition see \ref{other} and \ref{trd}.

\subsection{A sufficient condition}
Here we examine conditions under which equation \eqref{three} is valid. Obviously this obtains, if 
\begin{equation}\label{suff}\max_{j\in \brii\(x\)}\alpha\(x,j\)=\min_{j\in J}\alpha\(x,j\), \quad \forall x\in X\end{equation} 
Also, if there exists some $x_0\in X$ such that 
\begin{equation}\label{relax}\max_{j\in \brii\(x_0\)}\alpha\(x_0,j\)=\min_{j\in J}\alpha\(x_0,j\)\ge \max_{j \in \brii\(x\)}\alpha\(x,j\), \quad \forall x\in X,\end{equation}
then, the left hand side of \eqref{relax} is $\alpha^H$ and since $v_A=\max_{x \in X}\min_{j \in J}\alpha\(x,j\)\ge\min_{j\in J}\alpha\(x_0,j\)=\alpha^H$, we conclude that \eqref{relax} is sufficient to get $v_A=\alpha^L=\alpha^H$, i.e \eqref{three}.

\begin{example} 
The $2\times 2$ bimatrix game with payoff matrices
\[A=\begin{pmatrix*}[r] 1 & 0 \\ -2 & -10 \end{pmatrix*}\qquad B=\begin{pmatrix*}[r] 0 & 1\\ 1 & 0 \end{pmatrix*}\]
satisfies condition \eqref{relax}, but not \eqref{suff} in $\gi$. 
\end{example}

\subsection{Leader payoffs in wuc games}\label{wucg}
We now turn our attention to two person weakly unilaterally competitive (wuc) games. 
\begin{definition}\label{wucdef}
A bimatrix game $\Gamma$ is weakly unilaterally competitive, if for all $x_1, x_2 \in X$ and all $ y \in Y$
\begin{align*}
\alpha\(x_1, y\) > \alpha\(x_2, y\) & \implies \beta\(x_1, y\) \le \beta\(x_2, y\)\\ 
\alpha\(x_1, y\) = \alpha\(x_2, y\) & \implies \beta\(x_1, y\) =  \beta\(x_2, y\)\end{align*}
and similarly if for all $y_1, y_2 \in Y$ and all $x \in X$ 
\begin{align*}
\beta\(x, y_1\) > \beta\(x, y_2\) &\implies \alpha\(x, y_1\) \le  \alpha\(x, y_2\)\\
\beta\(x, y_1\) = \beta\(x, y_2\) &\implies \alpha\(x, y_1\)  =  \alpha\(x, y_2\) \end{align*}
\end{definition}
\cite{Be02} observes that the classes of two-person zero-sum, strictly competitive (sc), unilaterally competitive (uc) and weakly unilaterally competitive (wuc) bimatrix games satisfy the inclusion relation $\(\text{zero-sum}\) \subsetneq (\text{sc}) \subsetneq (\text{uc}) \subsetneq (\text{wuc})$.  As an immediate consequence of the definition, wuc games satisfy the sufficient condition \eqref{suff}, and thus the leader's equilibrium payoff is unique and both the commitment value and all Nash equilibria payoffs are equal to the matrix game value. Formally,  
\begin{proposition} \label{wuc}In a wuc game $\Gamma$, the leader's payoff at any subgame perfect equilibrium of the commitment game is equal to his commitment value which is equal to his matrix game value, i.e. $v_A=\alpha^L=\alpha^H$ and $v_B=\beta^L=\beta^H$. Moreover, all Nash equilibria payoffs of the wuc game are equal to the matrix game values, i.e. $(\alpha^N, \beta^N)=\(v_A, v_B\)$ for all $(x^N, y^N)$.
\end{proposition}
\begin{proof} By Definition \ref{wucdef}, we have that for any $x \in X$ and any $j_1,j_2,k \in J$ with $j_1, j_2 \in \brii\(x\)$ and $k\notin \brii\(x\)$
\begin{align*}
\beta\(x, j_1\) = \beta\(x, j_2\) &\implies \alpha\(x, j_1\) = \alpha\(x, j_2\)\\
\beta\(x, j_1\) > \beta\(x, k\phantom{'}\)&\implies \alpha\(x, j_1\) \le \alpha\(x, k\phantom{'}\)\end{align*}
which together imply that $\max_{j\in \brii\(x\)}\alpha\(x,j\)=\min_{j\in J}\alpha\(x,j\)$, i.e \eqref{suff} is satisfied. Similarly, using the first part of Definition \ref{wucdef}, the result follows for player II. The second part of the proposition, already known in the literature, follows trivially from the inclusion of Nash equilibria payoffs between the matrix game values and the highest subgame perfect equilibrium payoff $\alpha^H$ (resp. $\beta^H$) of the commitment game.\end{proof}

\Citet{Hu96} derive a result similar to our Proposition \ref{wuc} for the class of strictly competitive games, which, as we have seen, 
is a subclass of wuc games. \par

\emph{\underline{Note}}: Proposition \ref{wuc} can be generalized for N-player wuc games, $N>2$. Such games are defined similarly to Definition \ref{wucdef} (see \cite{Ka92}). Also, for N-player games, the corresponding leadership game $\Gamma^1$, is defined similarly: in the first stage, the leader (say, player 1) commits to a mixed strategy and in the second stage, the remaining $N-1$ players \emph{simultaneously} choose their strategies knowing the mixed strategy of the leader. So, take $\Gamma$ to be a wuc N-player game and let $\Gamma(x)$ be the sub-game where player 1 has fixed his mixed strategy $x$. Then, $\Gamma(x)$ is also a wuc game. Let  $Y=\prod_{i=2}^N Y_i$ be the set of mixed \emph{uncorrelated} strategy profiles of the followers and let $Y(x)$ be the set of strategies of the followers that are at equilibrium in $\Gamma(x)$. By a result of \cite{Wo99} for N player wuc games, if $\alpha(x,y)$ denotes player 1's payoff in $\Gamma(x)$, then, $\alpha\(x,y_1\)=\alpha\(x,y_2\)=\min_{y\in Y}$ for any $y^1,y^2 \in Y\(x\)$, i.e. for each $x$, player 1's payoff is constant over all equilibrium strategies of the followers in $\Gamma(x)$ and it is the worse possible outcome for him over all mixed strategies of the followers. But then, $\max_{y\in Y(x)}\alpha\(x,y\)=\min_{y\in Y}\alpha\(x,y\)$ which implies that $\max_{x \in X}\max_{y\in Y(x)}\alpha\(x,y\) = \max_{x \in X}\min_{y\in Y}\alpha\(x,y\)$.

\subsection{Games of common interest}\label{alt}
Opposite to bimatrix games resembling zero-sum games stand games where there is a strong motivation for the cooperation of the two players. In such games we may get the equality of $\alpha^L$ and $\alpha^H$ by a condition opposite to \eqref{suff}, namely
\begin{equation}\label{conv}\max_{j\in \brii\(x\)}\alpha\(x,j\)=\max_{j\in J}\alpha\(x,j\), \quad \forall x\in X\end{equation}
This condition guarantees that $\alpha^H=\max_{i,j}\alpha\(i,j\)$, since $\max_{x \in X}\max_{j\in J}\alpha\(x,j\)=\max_{i,j}\alpha\(i,j\)$ and therefore \eqref{conv} is sufficient to get $\alpha^L=\max_{i,j}\alpha\(i,j\)$ for non-degenerate games. For degenerate games though, this is not the case as the next example shows. 
\begin{example} 
The $2\times 2$ bimatrix game with payoff matrices
\[A=\begin{pmatrix*}[r] -1 & 2 \\ -2 & 0 \end{pmatrix*}\qquad B=\begin{pmatrix*}[r] 1 & 1\\ 1 & 1 \end{pmatrix*}\]
satisfies condition \eqref{conv} for player I, but $\alpha^L=-1<\alpha^H=2=\max_{ij}\alpha\(i,j\)$. 
\end{example}
An alternative sufficient condition for $\alpha^L$ to be equal to $\alpha^H$ is the existence of an $x_0\in X$ such that $\alpha^H$ is acquired at $x_0$ and $\alpha\(x_0,j\)$ is constant on $\brii\(x_0\)$. Then, $\min_{j\in \brii\(x_0\)}\alpha\(x_0,j\) \ge \max_{j\in \brii\(x\)}\alpha\(x,j\)$ for all $x\in X$, and hence $\alpha^L \geq \min_{j\in \brii\(x_0\)}\alpha\(x_0,j\) \geq \alpha^H$, which of course implies $\alpha^L=\alpha^H$. This condition is not sufficient to guarantee $\alpha^H=\max_{i,j}{\alpha\(i,j\)}$, as the following example shows
\begin{example} 
The $2\times 3$ bimatrix game with payoff matrices
\[A=\begin{pmatrix*}[r] 3 & 0 & 0\\ 0 & 2 & 4\end{pmatrix*}\qquad B=\begin{pmatrix*}[r] 3 & 0 & -1\\ 0 & 2 & -1 \end{pmatrix*}\]
satisfies the relaxed condition for player I, hence $\alpha^L=\alpha^H=3$, but $3<4=\max_{ij}\alpha\(i,j\)$. 
\end{example}

\subsection{Other generalizations of zero-sum games and some counterexamples}\label{other}
The result of Proposition \ref{wuc} does not apply to other generalizations of zero-sum games that have been studied in the literature. Such generalizations include the class of a-cooperative games, which is studied in a setting similar to ours by \cite{As80}, the class of almost strictly competitive games (asc) introduced by \cite{Au61}, the class of pre-tight games, \cite{Vi06}, and the class of strategically zero-sum games, \cite{Mo78}. We provide the main definitions of these classes for bimatrix games, but for the general case one is referred to the relevant works. \par

Let $\Gamma$ be a bimatrix game with strategy spaces $X,Y$ and payoff functions $\alpha, \beta$. A \textit{twisted equilibrium} of $\Gamma$ is a Nash equilibrium of the game $\tilde{\Gamma}$, which is played over the same strategy spaces $X,Y$ with payoff functions $\tilde{\alpha}:=-\beta$ and $\tilde{\beta}:=-\alpha$. We denote with $TE\(\Gamma\)$ and $TE\Pi\(\Gamma\)$ the sets of twisted equilibria and twisted equilibria payoffs of $\Gamma$. A pair of strategies $\(x^s,y^s\)$ is a called \textit{saddle-point} of $\Gamma$, if \[\alpha\(x,y^s\)\le \alpha\(x^s,y^s\)\le \alpha\(x^s,y\)\quad \text{and} \quad \beta\(x^s,y\)\le \beta\(x^s,y^s\)\le \beta\(x,y^s\)\] for all $x\in X, y\in Y$. We denote with $S\(\Gamma\)$ the set of saddle points of $\Gamma$. For every bimatrix game, $S\(\Gamma\)=NE\(\Gamma\)\cap TE\(\Gamma\)$. A pair of strategies $\(x,y\)$ is \textit{Pareto-optimal} if there is no other pair of strategies $\(x',y'\)$ giving at least as much to every player and more to some player. \par 

A bimatrix game $\Gamma$ is \textit{a-cooperative} if it has at least one Pareto-optimal twisted equilibrium. A bimatrix game $\Gamma$ is \textit{almost strictly competitive (asc)} if $S\(\Gamma\)\neq\emptyset$ and $NE\Pi\(\Gamma\)=TE\Pi\(\Gamma\)$. \par 

Similarly to wuc games, see section \ref{wucg}, \cite{As80} and \cite{Be02} observe that the classes of zero-sum, a-cooperative and asc bimatrix games satisfy the inclusion relation $\(\text{zero-sum}\) \subsetneq \(\text{a-cooperative}\) \subsetneq \(\text{asc}\)$. Over asc games, Nash equilibria payoffs of the simultaneous move game are constant and equal to the players' matrix game values, i.e. $(\alpha^N, \beta^N)=\(v_A, v_B\)$ for all $(x^N, y^N)$ of the simultaneous move game. \par

In any two-person asc game the unique Nash equilibrium payoff is also the unique twisted equilibrium payoff. \Citet{As80} use this property to prove that in any a-cooperative game, both leaders in the associated leadership games will commit to Nash equilibrium strategies of the simultaneous move game and consequently they will receive their matrix game values, i.e. for a-cooperative games \eqref{three} is true. \par

However, this property does not extend to the class of asc games. As shown below (see section \ref{trd}), the TrD is an asc game in which both players strictly improve their payoffs in the associated leadership games. In particular, \eqref{two} is true for TrD. \par

\Citet{Vi06} defines and studies the class of pre-tight games. A pure strategy $i$ (resp. j) of player I (resp.II) is called \textit{coherent} if it is played in a correlated equilibrium of $\Gamma$. A bimatrix game $\Gamma$ is called \textit{pre-tight} if in any correlated equilibrium all the incentive constraints for non deviating to a coherent strategy are tight. \Citet{Vi06} notes that the class of two-player pre-tight games strictly contains two player zero-sum games and games with a unique correlated equilibrium. TrD, has a unique correlated equilibrium and thus provides an example that in a pre-tight game, contrary to zero-sum games, the leader and the follower may strictly improve their payoffs in the associated leadership games $\gi$ and $\gii$. \par
\Citet{Mo78} define strategically zero-sum bimatrix games, which strictly contain zero-sum and strictly competitive games. A bimatrix game $\Gamma$ with payoff matrices $\(A,B\)$ is \textit{strategically zero-sum} if it is strategically equivalent to a zero-sum game, i.e. if there exists a matrix game with payoff functions $c,-c$ on the same strategy spaces, such that $\alpha\(x',y\)\ge\alpha\(x,y\)\iff c\(x',y\)\ge c\(x,y\)$ for all $x,x' \in X, y\in Y$ and $\beta\(x,y'\)\ge\beta\(x,y\)\iff c\(x,y'\)\le c\(x,y\)$ for all $x \in X, y,y'\in Y$. Additionally, they define the concept of a trivial game for a player and show that every $2\times 2$ bimatrix game that is trivial for a player is strategically zero-sum. They show that any strategically zero-sum game is best response equivalent\footnote{Two games are \textit{best response equivalent} if they have the same best response correspondence, see \cite{Ro74}.} to a zero-sum game and provide sufficient conditions for the converse to be true. They also show that in strategically zero-sum games a Nash equilibrium exists that dominates payoff-wise all other Nash equilibria of the simultaneous game. \par
Example \ref{conitzer} shows that in the class of strategically zero-sum games $\alpha^L$ may be strictly greater than the payoff of the leader at the dominant Nash equilibrium of the simultaneous game.
\begin{example}\label{conitzer}
Let $\Gamma$ be the $2\times 2$ bimatrix game
\[A=\begin{pmatrix*}[r] 2 & -1 \\ 3 & 0 \end{pmatrix*}\qquad B=\begin{pmatrix*}[r] 2 & 1\\ -1 & 0 \end{pmatrix*}\]
Then, strategy $s_1$ of player I is strictly dominated, thus the game is trivial for player I according to the triviality concept of 
\cite{Mo78}. The only Nash equilibrium of the simultaneous move game is $(x^N,y^N)=\(s_2,t_2\)$ with payoffs $(\alpha^N,\beta^N)=\(0,0\)$. 
In $\gi$, $\alpha^L=\alpha^H$ and player's I equilibrium strategy is $x^L=\(0.5,0.5\)$ which induces player II to play $t_3$. The resulting unique equilibrium yields payoffs $(\alpha^L,\beta^F)=\(2.5,0.5\)$. 
\end{example}

While best response equivalent games have the same set of Nash equilibria, the equilibria of their associated leadership games may differ. This is highlighted in the following example.
\begin{example}\label{31019} Let $\Gamma$ be the $2\times2$ bimatrix game
\[A=\begin{pmatrix*}[r] 3 & 10 \\ 1 & 9 \end{pmatrix*}\qquad B=\begin{pmatrix*}[r] 3 & 1\\ 8 & 9 \end{pmatrix*}\]
and $\Gamma'$ the bimatrix game 
\[A'=\begin{pmatrix*}[r] 3 & 2\\ 1 & 1\end{pmatrix*}\qquad B'=B\]
$\Gamma'$ results from $\Gamma$ by the affine transformation $\alpha'\(\cdot, t_4\)=\alpha\(\cdot, t_4\)-8$. $\Gamma$ and $\Gamma'$ are best response equivalent, since the best reply regions $Y\(i\), Y'\(i\)$ for $i=1,2$ and $X\(j\), X'\(j\)$ for $j=3,4$ are 
\[\begin{array}{ll}
Y\(1\)=Y'\(1\)=Y, & Y\(2\)=Y'\(2\)=\emptyset \\
X\(3\)=X '\(3\)=\left [\frac13,1\right ], & X\(4\)=X'\(4\)=\left [0, \frac13\right ]
\end{array}\]
and hence they have the same set of Nash equilibria, which is the singleton $\(s_1,t_3\)$ with payoffs $\(\alpha^N,\beta^N\)=\(3,3\)$. However, the unique leader equilibrium of $\gi$ is 
\[\(x^L,j^F\)=\(\(\frac13,\frac23\),t_4\), \quad \(\alpha^L,\beta^F\)=\(9\frac13, 6\frac13\)\]
while the unique leader equilibrium of $\Gamma'^{\text{I}}$ is 
\[\(x'^L,j'^F\)=\(s_1,t_3\), \quad \(\alpha'^L,\beta'^F\)=\(3,3\)\]
\end{example}

\subsection{Traveler's dilemma}\label{trd}
Traveler's dilemma (TrD) was first introduced by \cite{Ba94} and quickly attracted widespread attention as a game where rationality leads to a difficult to accept Nash equilibrium solution. In TrD, the equilibrium solutions of the simultaneous game and of the associated leadership games differ significantly. Notably, the associated leadership equilibria strategies are much closer to the behavior of the players, as observed in experiments, than the unique Nash equilibrium of the simultaneous move game. This discrepancy generates a strong motivation to study leadership games. \par
TrD is a symmetric bimatrix game with strategy spaces $I=J=\{2,3,\ldots,99,100\}$ and payoffs \[\alpha\(i,j\)=\beta\(j,i\)=\begin{cases}i+2,& i<j\\ i, & i=j\\ j-2,& i>j\end{cases}\]
The best reply correspondence of player I against $j\in J$ is \[\operatorname{BR^I}\(j\)=\begin{cases}j-1, &j>2\\ 2, & j=2\end{cases}\] and similarly for player II. The game can be solved with the process of iterated elimination of strongly dominated strategies (iesds) and hence it has the fictitious play property (see \cite{Mo96}). In the first round of elimination, the pure strategy $\(100\)$ is strongly dominated by the mixed strategy $x_{\epsilon}:=\(1-\epsilon\)\(99\)+\epsilon\(2\)$ (for $0<\epsilon<\frac1{97}$) and hence eliminated. The process of iesds successively deletes all strategies except the pure strategy $\(2\)$, which results to the pure strategy profile $\(x,y\)=\(2,2\)$ being the unique Nash equilibrium. Hence, the equilibrium $\(2,2\)$ survives any Nash equilibrium refinement concept. It is also the unique correlated equilibrium and the unique rationalizable strategy profile\footnote{See \cite{Be84} and \cite{Pe84} for the definition of rationalizable strategic behavior.}, implying that neither correlation nor rationalizability may improve upon this Nash equilibrium. The TrD has attracted interest mostly due to the fact that this unique solution, although having a very strong theoretical argument in its favor since it is derived by iesds, is inefficient in terms of the social welfare, counter-intuitive and differs significantly from the observed behavior of the players in conducted experiments. In the penultimate round of elimination the game corresponds to a prisoner's dilemma
\[A=\begin{pmatrix*}[r]3 & 0\\ 4 & 2 \end{pmatrix*}\qquad B=\begin{pmatrix*}[r]3 & 4\\ 0 & 2 \end{pmatrix*}\]
Hence, TrD may be viewed as a generalization of prisoner's dilemma to $n$ strategies. Contrary to the intuition that cooperation should be beneficial for both players, this is not the case as TrD turns out to be an almost strictly competitive game, since
\begin{enumerate}[leftmargin=0cm,itemindent=.5cm, labelwidth=\itemindent,labelsep=0cm, noitemsep, align=left]
\item $S\(\Gamma\)=\{\(2,2\)\}$
\item $NE\Pi\(\Gamma\)=TE\Pi\(\Gamma\)=\{2,2\}$
\end{enumerate}
a fact that forces the players' payoffs to their matrix game values. TrD resembles Bertrand duopoly, an observation made already in \cite{Ba94}. \Citet{Ha12} apply their new solution concept of iterated regret minimization to TrD and derive a satisfactory solution. \par

Although correlated equilibrium or rationalizability do not improve upon the Nash equilibrium outcome in TrD, coarse correlation and leadership significantly do so. It is straightforward to check that the distribution $\(z_{ij}\)_{\(i,j\)\in I\times J}$ with 
\[z_{ij}=\begin{cases}\frac12, & i=j=100, \text{ and } i=j=98 \(\text{ or } i=j=97\),\\0, & \text{ else }\end{cases}\]
is a symmetric coarse correlated equilibrium, with payoffs equal to $99$ (or $98.5$) for each player. Note that starting from this distribution, one may see that there exists a great multitude of coarse correlated equilibria in TrD. However, the Pareto-optimal outcome $\(100,100\)$ is not a coarse correlated equilibrium. \par

The unique equilibrium of $\gi$ (commitment optimal strategy of the leader and optimal response of the follower) may be calculated by equation \eqref{alphaH} after considerable simplifications in the strategy spaces and is given by \[x^L=\frac13\(100\)+\frac13\(99\)+\frac13\(97\), \quad j^F=\(99\)\] with payoffs $\alpha^L=\beta^F=98\frac13$. The leader-follower approach to a solution concept works well for TrD since payoffs are the same for both players, irrespective of who moves first, and they are very close to the Pareto optimal outcome. In other words, if the rules of the game permit commitment, one expects the leader-follower equilibrium to prevail over the Nash equilibrium of the simultaneous move game. 

\section{\texorpdfstring{$2\times2$}{j} bimatrix games}\label{twobim}
In bimatrix games, where $I=\{1,2\}$ and $J=\{3,4\}$, we may derive some special properties for the associated leadership games. Let the payoff matrices $\(A,B\)$ be
\[A=\begin{pmatrix*}[r] a_1 & a_2\\ a_3 & a_4\end{pmatrix*}\qquad B=\begin{pmatrix*}[r]b_1 & b_2\\ b_3 & b_4 \end{pmatrix*}\]
In this case, when solving $\gi$, the possible edges of the best reply regions $X\(3\)$ and $X\(4\)$ are $C\(3\)\cup C\(4\)=\left\{s_1,x^d,s_2\right\}$, where $x^d:=\(1-d,d\)$ denotes the equalizing strategy of player I over player II's payoffs with 
\[d:=\frac{b_1-b_2}{b_1-b_2+b_4-b_3},  \quad \text{if }b_1-b_2+b_4-b_3\neq 0 \,\, \text{ and }\, 0\le d\le 1\,.\] 
Similarly, in $\gii$ the possible edges of the best reply regions $Y\(1\)$ and $Y\(2\)$ are $C\(1\)\cup C\(2\)=\left\{t_3,y^c,t_4\right\}$, where $y^c:=\(1-c,c\)$ denotes the equalizing strategy of player II over player I's payoffs, i.e. \[c:=\frac{a_4-a_3}{a_1-a_2+a_4-a_3}, \quad \text{if } a_1-a_2+a_4-a_3\neq 0 \,\, \text{ and }\, 0\le c\le 1\,.\] 
For a $2\times 2$ bimatrix game that is degenerate for the leader\footnote{Without loss of generality, we henceforth assume that the leader is player I.}, we first show that a stronger statement than that of Proposition \ref{pure} holds, namely that \textit{any} commitment optimal strategy $x^L$ is a Nash equilibrium strategy of the simultaneous move game.   
\begin{lemma}\label{alldeg}
If a $2\times 2$ bimatrix game $\Gamma$ is degenerate for the leader, then $X^L\subseteq NE\(X\)$. 
\end{lemma}
\begin{proof} If the game is degenerate for player I, then player II has either a weakly dominated strategy or two payoff equivalent strategies. \par 
In the first case, assume the weakly dominated strategy is $t_4$. This implies $X\(3\)=X, X^{\mathrm{o}}\(4\)=\emptyset, D=\{t_3\}$ and hence by equation \eqref{alphaL}, \[\alpha^L=\max_{j\in D}\left\{\max_{x\in X\(j\)}\min_{k\in \mathcal{E}\(j\)}\alpha\(x,k\)\right\}=\max_{x\in X\(3\)}\alpha\(x,t_3\)=\max_{x\in X}\alpha\(x,t_3\)\] which implies that $X^L=\argmax_{x\in X} \alpha\(x,t_3\)=\bri\(t_3\)$. Thus, any strategy that guarantees player I his payoff $\alpha^L$ in $\gi$ is a best response against $t_3$, which together with $X\(3\)=X$ implies that $\(x^L,t_3\)$ is a Nash equilibrium in the simultaneous move game for all $x^L \in X^L$, i.e. that $X^L\subseteq NE\(X\)$. \par 

On the other hand, if $t_3$ and $t_4$ are payoff equivalent, then $X\(3\)=X\(4\)=X$ and $\mathcal E \(t_3\)=\{t_3,t_4\}$. By equation \eqref{alphaL}, $\alpha^L=\max_{j\in J}\left\{\max_{x\in X}\min_{j\in J}\alpha\(x,j\)\right\}=v_A$. Obviously, any strategy $x^L\in X^L$ that guarantees player I his safety level is not a strongly dominated strategy. Since strongly dominated strategies and strategies that are never best responses coincide in bimatrix games, any $x^L \in X^L$ must be a best response against some $y\in Y$, which together with $Y=\brii\(x\)$ for any $x\in X$, implies $X^L\subseteq NE\(X\)$. 
\end{proof}
The statement of Lemma \ref{alldeg} is not true if we consider $x^H$ instead of $x^L$.
\begin{example}\label{4231} The $2\times 2$ bimatrix game $\Gamma$ with payoff matrices
\[A=\begin{pmatrix*}[r] 4 & 2 \\ 3 & 1\end{pmatrix*}\qquad B=\begin{pmatrix*}[r] 1 & 2\\ 0 & 0 \end{pmatrix*}\]
has a unique Nash equilibrium which is given by $\(x^N,y^N\)=\(s_1,t_4\)$, with payoffs $\(\alpha^N,\beta^N\)=\(2,2\)$. This is also $\alpha^L$ for player I in $\gi$. However, $x^H=s_2$ with payoff equal to $\alpha^H=3$, where $s_2$ is obviously not a Nash equilibrium strategy as it is strongly dominated by $s_1$. 
\end{example}
For a $2\times 2$ bimatrix game that is non-degenerate for the leader, we may partially strengthen the statement of Proposition \ref{pure}. In these games, commitment optimal strategies are Nash equilibrium strategies provided that no pure strategy of the leader is strongly dominated.
\begin{lemma}\label{allndeg}
If a $2\times 2$ bimatrix game $\Gamma$ is non-degenerate for the leader and no pure strategy of the leader is strongly dominated, then $X^L \subseteq NE\(X\)$.
\end{lemma}
\begin{proof} If player II has a strongly dominated strategy, say $t_4$, then $X\(3\)=X, X\(4\)=\emptyset$, $\alpha^L=\max_{x\in X}{\alpha\(x,t_3\)}$ and the result follows as in the previous Lemma. If player II has no strongly dominated strategy, then non-degeneracy implies that the payoffs of $B$ satisfy $b_1>b_2$ and $b_3<b_4$ $($if necessary by rearranging $B)$. Hence, $X\(3\)=\left[s_1,x^d\right]$ and $X\(4\)=\left[x^d,s_2\right]$, and by equation \eqref{alphaL}, $\displaystyle \alpha^L=\max{\left\{\max_{x\in \,\left[s_1,x^d\right]}{\alpha\(x,t_3\)},\max_{x\in\, \left[x^d,s_2\right]}{\alpha\(x,t_4\)}\right\}}$. If $X^L\subseteq I$, then $X^L\subseteq NE\(X\)$ by Proposition \ref{pure}. If $X^L=X$, then $\alpha\(x, y\)$ must be constant over $X,Y$ and the assertion holds trivially. If $X^L=\left[s_1,x^d\right]$, then $a_1=a_3>a_4$, hence $\bri\(t_3\)=X$ and the assertion follows. Similarly, if $X^L=\left[x^d,s_2\right]$. Finally, if $X^L=\left\{x^d\right\}$, then the pair of equalizing strategies $\(x^d,y^c\)$ is a Nash equilibrium of $\Gamma$ and the assertion follows. Notice that in this last case the existence of $y^c$ is guaranteed by the assumption that player I has no strongly dominated strategy.
\end{proof}
The next Lemma provides a necessary and sufficient condition for a commitment optimal strategy of a player \emph{not} to be a Nash equilibrium strategy in a non-degenerate $2 \times 2$ bimatrix game. If the condition obtains, then all commitment optimal strategies of the other player are Nash equilibrium strategies.
\begin{lemma}\label{domin}
If a $2\times 2$ bimatrix game $\Gamma$ is non-degenerate for the leader, then he may obtain his commitment value $\alpha^L$ using a strategy $x^L \notin NE\(X\)$ if and only if he has both
\begin{enumerate}[label=$\(\ell\arabic*\)$, leftmargin=0.2cm, topsep=0cm, itemindent=.5cm, labelwidth=\itemindent, noitemsep, align=left]
\item a strongly dominated strategy in $\Gamma$ and 
\item \sloppy an equalizing strategy $x^d=\(1-d,d\)$ over the follower's payoffs, such that $\alpha\(x^d, j\)\ge \alpha^N$ for some $j\in J$.
\end{enumerate}
Then, in game $\gii$, $Y^L= NE\(Y\)$.
\end{lemma}
\begin{proof}
Let $\(\ell 1\)$ and $\(\ell 2\)$ be true and take $s_1$ to be the strongly dominated strategy of the leader, i.e. $a_3>a_1$ and $a_4>a_2$. Since player I has an equalizing strategy $x^d$ over player II's payoffs, player II has no strongly dominated strategy. Since the game is non-degenerate for the leader, player II has no weakly dominated strategy. So, by rearranging if necessary, $b_1>b_2$ and $b_3<b_4$. This game has a unique Nash equilibrium which is the pure strategy profile $\(s_2,t_4\)$ with payoffs $\(\alpha^N,\beta^N\)=\(a_4,b_4\)$, hence $NE\(X\)=\left\{s_2\right\}$. The edges of the best reply regions of player II are: $C\(3\)\cup C\(4\)=\left\{s_1,x^d,s_2\right\}$ with corresponding payoffs for player I 
\begin{align*}\alpha\(x, t_3\)_{x\in C\(3\)}&=\left\{a_1,\(1-d\)a_1+da_3\right\}=\left\{a_1,a_1+d\(a_3-a_1\)\right\}
\\ \alpha\(x, t_4\)_{x\in C\(4\)}&=\{\(1-d\)a_2+da_4,a_4\}=\left\{a_4+\(1-d\)\(a_2-a_4\),a_4\right\}\end{align*}
Now, $a_3>a_1$ implies $a_1+d\(a_3-a_1\)>a_1$ and $a_4>a_2$ implies $a_4+\(1-d\)\(a_2-a_4\)<a_4$, hence $\alpha^L=\max{\left\{\alpha\(x^d, t_3\), \alpha\(s_2,t_4\)\right\}}=\max{\left\{a_1+d\(a_3-a_1\), a_4\right\}}$. Therefore, by $\(\ell 2\)$, $x^d\in X^L$. \par
For the converse, if $x^L \notin NE\(X\)$, by Proposition \ref{pure}, $x^L$ may not be a pure strategy. By Lemma \ref{allndeg}, player I must have a strongly dominated strategy. But then, proceeding as above, we conclude that $\(\ell 2\)$ must also be true. \par
Finally, to show that $Y^L= NE\(Y\)$ in $\gii$, observe that since player I has a strongly dominated strategy, if the roles of the players are reversed and II becomes leader, then he obtains $\beta^L$ in $\gii$ by using any of his Nash equilibrium strategies of the simultaneous move game.
\end{proof}
Combining the results of Lemmas \ref{alldeg}, \ref{allndeg}, and \ref{domin}, we obtain a clear picture of the relationship between Nash equilibria of the simultaneous move game $\Gamma$ and commitment optimal strategies in the leader-follower games $\gi$ and $\gii$. 
\begin{proposition}\label{twobytwo}
In a $2\times 2$ bimatrix game $\Gamma$: 
(a) If either $\gi$ is degenerate for player \I or $\gii$ is degenerate for player \II, then both $X^L\subseteq NE\(X\)$ and $Y^L\subseteq NE\(Y\)$. 
(b) If both $\gi$ is non-degenerate for player \I and $\gii$ is non-degenerate for player \II, then either 
(i) in $\gi$ there exists $x^L \notin NE\(X\)$ and in $\gii$, $Y^L=NE\(Y\)$ or 
(ii) in $\gii$ there exists $y^L \notin NE\(Y\)$ and in $\gi$, $X^L=NE\(X\)$ or 
(iii) $X^L\subseteq NE\(X\)$ and $Y^L\subseteq NE\(Y\)$.
\end{proposition}
\begin{proof}
Using Lemma \ref{alldeg}, for part (a) it suffices to show that if $\Gamma$ is degenerate for player I, then $Y^L\subseteq NE\(Y\)$. To show this, it is sufficient to show that condition $\(\ell 1\)$ of Lemma \ref{domin} never obtains when player II is a leader (i.e in $\gii$). This is true since the assumption that  the game is degenerate for player I implies that II has no strictly dominated strategy (II has either a weakly dominated strategy or his two strategies are payoff equivalent). For part (b), it suffices to show that conditions $\(\ell 1\)$ and $\(\ell 2\)$ may not obtain at the same time for $\gi$ and $\gii$. Indeed, if player I has a strictly dominated strategy (i.e. $\(\ell 1\)$ is true for player I in $\gi$), then player II may may not have an equalizing strategy over player I's payoffs (i.e. $\(\ell 2\)$ may not be true for player II in $\gii$), and if player I has an equalizing strategy over player II's payoffs in $\gi$, then player II may not have a strictly dominated strategy in $\Gamma$.
\end{proof}
A statement related to Proposition \ref{twobytwo} under a specific context can be found in \cite{St16}. Proposition \ref{twobytwo} does not imply that the commitment value $\alpha^L$ coincides with the leader's Nash equilibrium payoff in the simultaneous move game when $X^L\subseteq NE\(X\)$. It may be the case that $x^L$ induces the follower to use the most favourable best reply for the leader. In that case, the leader will get a higher payoff in $\gi$ than in $\Gamma$.
\begin{example}\label{4132}The $2\times 2$ bimatrix game $\Gamma$ with payoff matrices 
\[A=\begin{pmatrix*}[r] 4 & 1 \\ 3 & 2\end{pmatrix*}\qquad B=\begin{pmatrix*}[r] 1 & 2\\ 4 & 3 \end{pmatrix*}\]
is non-degenerate and hence the leader's equilibrium payoff is unique in both associated leadership games $\gi$ and $\gii$. The unique Nash equilibrium of the simultaneous move game is $\(x^N,y^N\)=\(\(0.5,0.5\),\(0.5,0.5\)\)$ with payoffs $\(\alpha^N,\beta^N\)=\(2.5,2.5\)$. In $\gi$, the leader commits to his Nash equilibrium strategy, i.e. $x^L=x^N$, however he may induce the follower to use $t_3$ and hence $\alpha^L=3.5>\alpha^N=2.5$, despite the fact that $x^L=x^N$. The same is true for player II in $\gii$. 
\end{example}
The statement of Proposition \ref{twobytwo} does not hold in $2\times n$ bimatrix games with $n>2$. 
\begin{example}
The $2\times3$ bimatrix game $\Gamma$
\[A=\begin{pmatrix*}[r]3 & 2 & 0\\ 1 & 4 & -1\end{pmatrix*}\qquad B=\begin{pmatrix*}[r] -2 & -1 & 0\\ 3 & -2 & 2\end{pmatrix*}\]
has a unique Nash equilibrium $\(x^N,y^N\)=\(s_1, t_5\)$ with payoffs $\(\alpha^N,\beta^N\)=\(0, 0\)$. Since $\(A,B\)$ is non-degenerate, the leaders' equilibria payoffs are unique in $\gi$ and $\gii$. In $\gi$, $\(x^L,j^F\)=\(\(\frac13,\frac23\),t_3\)$ with $\(\alpha^L,\beta^F\)=\left(\frac53,\frac43\right)$ and $x^L \notin NE\(X\)$. Similarly in $\gii$, $\(y^L,i^F\)=\(\(0.5,0.5,0\),s_1\)$ with $\(\beta^L,\alpha^F\)=\left(0.5,2.5\right)$ and $y^L \notin NE\(Y\)$. 
\end{example}

\subsection{Follower payoffs in non-degenerate \texorpdfstring{$2\times 2$ bimatrix games}{k}}\label{follower}
\Citet{St04} provide an example of a parametric $3\times3$ bimatrix game which shows that the follower's payoff may be worse or better than his Nash equilibrium payoff in the simultaneous move game. For the $2\times 2$ bimatrix case, one may use Proposition \ref{twobytwo} and derive conditions that determine the relationship between the follower's equilibrium payoff in $\gi$ and his Nash equilibria payoffs in $\Gamma$. The derivation of similar conditions for higher dimension cases is open. \par 
So, consider a $2\times 2$ non-degenerate bimatrix game $\Gamma$ and recall from section \ref{properties} that for each $x^L\in X^L$, the corresponding follower's payoff is $\beta^F:=\beta(x^L,j^F)$, where $j^F \in J$ such that $\alpha(x^L, j^F)=\alpha^L$. Since $j^F \in \brii(x^L)$, we have $\beta^F=\max_{j\in J}\beta(x^L,j)\ge \min_{x\in X}\max_{j\in J}\beta\(x,j\)=v_B$, i.e. any equilibrium payoff of the follower in $\gi$ is at least as high as his matrix game value (obviously, this is true for any bimatrix game and not just for $2\times2$). \par
If $x^L\in NE\(X\)$, then the follower's payoff $\beta^F$ against $x^L$ in $\gi$ is equal to his payoff in some Nash equilibrium of $\Gamma$ (the one of which $x^L$ is a component). On the other hand, if there exists $x^L\notin NE\(X\)$, then the corresponding follower's payoff $\beta^F$ may be lower than any $\beta^N\in NE\Pi\(Y\)$. However, by Lemmas \ref{allndeg} and \ref{domin}, this may occur only under conditions $\(\ell1\)$ and $\(\ell2\)$.
\begin{proposition}\label{followerp}
In a non-degenerate $2\times 2$ bimatrix game $\Gamma$, if conditions $\(\ell1\)$ and $\(\ell2\)$ hold for player \I in $\gi$, then there exists an equilibrium payoff $\beta^F$ of the follower in $\gi$ that is lower than his unique Nash equilibrium payoff in $\Gamma$ (i.e. $\beta^F<\beta^N$) if and only if $v_B<\beta^N$. If $\(\ell2\)$ holds with strict inequality, then this $\beta^F$ is the unique equilibrium payoff of the follower.
\end{proposition}
\begin{proof}By Lemma \ref{domin}, in a $2\times 2$ non-degenerate bimatrix game, there exists $x^L\notin NE\(X\)$ if and only if conditions $\(\ell1\)$ and $\(\ell2\)$ hold. In this case, as shown in the proof of Lemma \ref{domin}, one may assume -- if necessary by rearranging the order of strategies in $I,J$ -- that $a_3>a_1, a_4>a_2$ and $b_1>b_2,b_3<b_4$ and  the only Nash equilibrium of $\Gamma$ is the strategy pair $\(s_2,t_4\)$ with payoffs $\(a_4,b_4\)$. The safety level $v_B$ of player II in $\Gamma$ is given by \[v_B=\min_{x\in X}\max_{j\in J}{\beta\(x,j\)}=\min{\left\{b_1,\beta^d, b_4\right\}}\] where $\beta^d$ denotes player II's payoff against player I's equalizing strategy, which is constant over $y\in Y$ and equal to $\beta^d=\frac{\det\(B\)}{b_1+b_4-b_2-b_3}$. Additionally, $\(\ell2\)$ implies that $x^d\in X^L$, (if $\(\ell2\)$ holds strictly, $X^L=\left\{x^d\right\}$). In any case, $x^d\notin NE\(X\)$. The payoff $\beta^F$of the follower against $x^d$ is equal to $\beta^d$. Hence, $\beta^F<\beta^N$ implies $v_B<\beta^N$. For the other direction, if $v_B<\beta^N \(=b_4\)$, then either $\beta^d<b_4$ (and hence $\beta^F<\beta^N$) or $b_1<b_4$. In the latter case, since $\beta^F=\beta^d=\beta\(x^d,t_3\)=b_1+d\(b_3-b_1\)$, $b_1<b_4$, $b_3<b_4$, and $b_4=\beta^N$, we again conclude that $\beta^F<\beta^N$.
\end{proof}
If both players are better off in $\gi$ or in $\gii$ than in any Nash equilibrium of the simultaneous move game, we may reason that they will have a strong incentive to play the game sequentially with the specified order, casting doubt to the Nash equilibrium as a possible outcome/solution concept for the game. 

\Urlmuskip=0mu plus 1mu
\bibliographystyle{plainnat}\bibliography{biblf}
\end{document}